\documentclass[12pt]{article}
\usepackage{amsmath}
\usepackage{graphicx}
\usepackage{float}
\usepackage{verbatim}
\usepackage{amsthm}
\usepackage{amsfonts}
\usepackage{amssymb}
\usepackage{listings}
\usepackage{fullpage}

\newtheorem{theorem}{Theorem}

\newtheorem{remark}[theorem]{Remark}
\newtheorem{mydef}[theorem]{Definition}
\newtheorem{lemma}[theorem]{Lemma}
\newtheorem{claim}[theorem]{Claim}
\newtheorem{corollary}[theorem]{Corollary}

\newcommand{\Ta}{\Theta}
\newcommand{\Om}{\Omega}
\newcommand{\de}{\delta}
\newcommand{\ve}{\varepsilon}
\newcommand{\eat}[1]{{}}

\usepackage{microtype}

\begin{document}

\title{Search using queries on indistinguishable items}

\author{Mark Braverman\thanks{Princeton University, research partially supported by an Alfred P. Sloan Fellowship, an NSF CAREER award, and a Turing Centenary Fellowship. }\and
Gal Oshri \thanks{Princeton University}}

\maketitle

\begin{abstract}
We investigate the problem of determining a set $S$ of $k$ indistinguishable integers in the range $[1, n]$. 
The algorithm is allowed to query an integer $q\in [1,n]$, and receive a response comparing this integer to an integer randomly chosen from $S$. 
The algorithm has no control over which element of $S$ the query $q$ is compared to. 
We show tight bounds for this problem. In particular, we show that in the natural regime where $k\le n$, the optimal number of queries to attain $n^{-\Omega(1)}$ error probability is $\Ta(k^3 \log n)$. 
In the regime where $k>n$, the optimal number of queries is $\Ta(n^2 k \log n)$. 

Our main technical tools include the use of information theory to derive the lower bounds, and the application of noisy binary search 
in the spirit of Feige, Raghavan, Peleg, and Upfal (1994). In particular, our lower bound technique is likely to be applicable in other 
situations that involve search under uncertainty. 
\end{abstract}

\newpage

\setcounter{page}{1}

\section{Introduction}
This paper investigates the problem of identifying a set $S$ of indistinguishable items by repeated queries where we know the range of values the items can take. At every query, we gain information based on our query and some random item from the set $S$ we are trying to find (we do not know which item was chosen). The overall simple statement of the problem makes it widely generalizable. The query can be thought of as an experiment in which we apply a measurement on an element of $S$ without knowing which element has been measured. The set of items can refer to a set of DNA strands in a ``soup'' of DNAs, passwords or any item that we might be interested in finding when we know what possible values the item may take. The queries can be viewed as tests on DNA strands, attempts at guessing a password or any trial we may run that will provide some information about one of the items in question. The specific problem we investigate is where the items are integers. Our queries are guesses of integers which return the result of a comparison with a chosen integer from the set we are trying to find.

As far as we know, this problem has not been investigated in the literature. However, it falls into the rich class of noisy search problems. Since we do not know which number was chosen when we query a number, we have to deal with a lack of information in trying to determine the set of numbers. Due to this missing information, it is not immediately obvious that there exists a solution to the problem. 

In this paper we give asymptotically tight upper and lower bounds for the number of queries needed to find a set $S$ of size $k$ of numbers from $\{1,\ldots,n\}$, where the queries are comparison queries. 

We briefly discuss similar problems that have been previously studied. Feige et al. explored the depth of noisy decision trees, where each node can be wrong with some constant probability, in \cite{Feige}. One of the problems they investigated is binary search where the result of each query is wrong with a constant probability. They presented an algorithm to solve this with running time $\Theta(\log{\frac{n}{Q}})$ where $n$ is the input set size and $Q$ is the probability of error of the algorithm. The algorithm we present uses a similar technique to the one used for noisy binary search in \cite{Feige}. 

The Renyi-Ulam game is also a related problem. In one variation of this game, we need to discover a chosen integer. To do this, we query a number and are told whether the number we are trying to find is greater than the number we guessed or not. However, some constant number of lies are allowed. In \cite{Spencer}, one lie is allowed, which means that one of the responses to our queries can be false. Similarly, Pelc discussed in \cite{Pelc} an algorithm for performing the search when one lie is allowed and concluded that the original question posed by Ulam (finding an integer between one and a million with one lie allowed) requires 25 queries. In \cite{Spencer}, \cite{Pelc} and other papers that explore the Renyi-Ulam game, some restriction is placed on the pattern of queries with false results. Ravikumar and Lakshmanan discussed such patterns (and why they are necessary to make the problem solvable) in \cite{Ravikumar}.

Another related problem is sorting from noisy information. Braverman and Mossel investigated this in \cite{Braverman}. The problem of sorting from noisy information is similar to our problem because in noisy sorting we can make comparisons between the items that need to be sorted, but each comparison may give us false information. This has applications, for example, in ranking sports teams where the comparisons are games between teams (one team wins) but the comparisons are noisy because the better team (which should have a higher rank) does not always win. Klein et al. also investigated this problem in \cite{Klein}. Apart from noisy sorting, they applied the same model to explore other problems, such as finding the maximum of $n$ numbers.


The problem we are investigating is motivated by applications that involve a search for several items by repeated queries where we do not know which item was chosen to be compared with our query (i.e. the items are indistinguishable). One interpretation is where the items represent DNA strands in a mixture that we are trying to identify. We can perform tests that give us some information about one of the DNA strands in the mixture, but we do not know which one. Similarly, instead of trying to identify DNA strands, we might be trying to identify passwords where our queries give us some partial information about one password out of several that a particular user often uses (and switches between). 

We note that the applications mentioned do not take the exact form as the problem we explore. The items in our problem are integers and the queries are guesses of an integer that result in the response `less than or equal to' or `greater than'. In generalizing the problem to other applications, the form of items or queries may change. For example, the queries in the DNA mixture example may describe a property of a particular nucleotide instead of returning one of two possible answers. Therefore, the algorithm will have to be changed. However, a similar framework can be used which allows information to be gained despite the uncertainty regarding query responses due to the indistinguishability of the items. A solution to the problem we have posed can lead to the development of new methods for identifying a set of items where we know these items can only take on a certain range of values.
On the lower-bound side, our results show that information-theoretic quantities are very effective at measuring and upper-bounding information learned from queries, even when such information is only a fraction of one bit. 
We believe that the information-theoretic lower bound technique will generalize to tight lower bounds in other settings. 

We now discuss the results and structure of the paper. In Section \ref{sec:description}, we formally introduce the problem we are solving with the restriction that the number of chosen integers is significantly smaller than the range of integers available. We prove a lower bound for the problem in Section \ref{sec:lowerbound} using information theoretic techniques. This involves constructing the hard instances where we split the possible values the chosen integers can take into consecutive clusters of equal size and place one chosen integer in each such cluster. Intuitively, this forces the search algorithm to find the elements one at a time, which turns out to be costly due to the fact that we don't control the sample. To formalize this intuition, we calculate the entropy of the random variable representing a particular chosen integer (it may take values of the integers in one of the clusters described above). We then use the mutual information of this random variable and the random variable representing the responses to the queries we make to find the minimum number of queries required to find that chosen integer. After showing that the same minimum number of queries applies to at least half of the chosen integers, we reach a lower bound of $\Omega(k^3 \log{n})$, where $k$ is the size of the set $S$ and the elements of $S$  take integer values between $1$ and $n$ (inclusive). Further, this bound extends to all $k<n$, using a slightly different set of hard instances. When $k>n$ we obtain a lower bound of $\Om(k^2 n\log n)$.  In Section \ref{sec:algorithm}, we present an optimal algorithm for solving the problem, proving both its correctness and worst case running time of $O(k^3 \log{\frac{n}{\delta}})$ where $\delta$ is the probability of error. This shows that the lower bound is tight. Moreover, while the lower bound applies to finding $S$ even with a constant error probability, we see that the upper bound remains asymptotically the same even if we set the error $\de=n^{-O(1)}$ to be polynomially small.  

Our results show that the problem we describe can be solved in practice when the items we are searching for can take a large number of values. This is because the dependence of the running time on $n$ grows as $\log{n}$. However, the number of items in $S$ needs to remain small because the dependence of the running time on $k$ grows as $k^3$.

\section{Problem definition} \label{sec:description}
We consider a (multi-)set $S$ of $k$ distinct integers where each is $X_i \in \{1, 2, \ldots, n\}$ for $1 \le i \le k$. Our goal is to discover the set $S$. The process is to repeat the following three steps: 
\begin{enumerate}
	\item Query an integer $Y \in \{1, 2, \ldots, n\}$.
	\item An integer $X_i$ is selected from $S$ uniformly at random.
	\item We are told whether $X_i \leq Y$ or $X_i > Y$.
\end{enumerate}
These three steps are repeated until we know what the $k$ integers in $S$ are. Our goal is to find the most efficient algorithm for determining $S$. Our model of computation is that queries are the costly operations. Therefore, by finding the most efficient algorithm we mean finding the algorithm that minimizes the number of queries made. We refer to this as `the problem' we are solving. Furthermore, for brevity, we refer to the two possible responses to queries as `$\leq$' ($X_i \leq Y$) and `$>$' ($X_i > Y$) and the $k$ integers in $S$ as `the chosen integers'. 

In this paper we give a complete characterization of the query complexity of this problem. 
Note that since the $X_i$ is selected at random from $S$, we cannot hope for a deterministic algorithm, and have to settle for a probabilistic performance guarantee. 
We focus on the regime where we are required to output the correct set $S$ except with some (possibly constant) probability $\de$. 
The answer can be broken down into three main regimes, which will 
be discussed in the analysis: (1) $k\ll n$, e.g. $k<\sqrt{n}$; (2) $\sqrt{n} < k < n$; and (3) $k\ge n$. 
The answer is given by the following main theorem:

\begin{theorem}
\label{thm:main1}
The number of queries needed to determine a multi-set $S\subset [n]$ of size $k$ with a given error $n^{-O(1)}<\de<1/4$ is $\Ta(k^3\log n)$ when $k\le n$, and $\Ta(k^2 n \log n)$ when $k\ge n$. 
\end{theorem}

Note that the distinction between $k<\sqrt{n}$ and $\sqrt{n}<k<n$ only comes up in the analysis, but (asymptotically) makes no difference in the result. 

\begin{remark}
Because of the way the algorithms work, Theorem~\ref{thm:main1} remains true even if the comparisons in the query answers are themselves noisy, 
and output the correct value of $X_i\stackrel{?}{>}Y$ correctly only with probability $1/2+\gamma$ for some constant $\gamma>0$. 
\end{remark}

\begin{remark}
Somewhat surprisingly, same bounds hold for a fairly broad range of error parameters. In particular, the lower bound holds even when the error is constant, while the upper bound
holds even for polynomially small errors (the constant in the $\Ta(\cdot)$ may depend on the constant $\beta$ in $\delta = n^{-\beta}$). 
\end{remark}

\section{The lower bounds}


We begin with showing the lower bound. In fact, we break the lower bound into two regimes: $k\le \sqrt{n}$ and $k>\sqrt{n}$. In 
the former regime, we use information-theoretic techniques to show the lower bound. In the latter, we give a more straightforward
proof of the $\Omega(k^3\log k)$ lower bound when $k<n$, and $\Om(k^2 n \log n)$ when $k>n$. The  $\Omega(k^3\log k)$
lower bound is weaker in general than  $\Omega(k^3\log n)$ when $k<n$, but is equivalent in the regime where $k>\sqrt{n}$.

\subsection{The case $k\le \sqrt{n}$: an information-theoretic lower bound} \label{sec:lowerbound}

The main technical ingredient in the lower bound proof is the Kullback-Leibler divergence and mutual information. We first introduce these terms and the lemmas we will use. For a more thorough introduction to these, see \cite{Cover}.

The Kullback-Leibler divergence (KL-divergence) measures the difference between two probability distributions:
\begin{mydef} 
For discrete random variables $P$ and $Q$ over sample space $\Omega$, the KL-Divergence is defined as: $$D_{KL}(P||Q) = \sum_{i \in \Omega} P(i) \log{\frac{P(i)}{Q(i)}}$$ with the convention that the term in the sum is interpreted as 0 when $P(i) = 0$ and $+\infty$ when $P(i) > 0$ and $Q(i) = 0$ 
\end{mydef}

We also use mutual information, which we define and arrange into a form we will use:
\begin{mydef}
Mutual information is a measure of the correlation between two random variables. The more independent the variables are, the lower the mutual information is. 
$$I(X;Y) = D_{KL}(p(x,y)||p(x)p(y))$$
\end{mydef}

Before we rearrange this definition into a form we will use, we first note (from \cite{Cover}) that it can also be written in terms of the more familiar Shannon entropy as:
$$I(X;Y) = H(X) - H(X|Y).$$
Since $H(X) \geq H(X|Y)$, $I(X;Y) \geq 0$. If entropy is interpreted as the uncertainty regarding a probability distribution, we see that the mutual information between $X$ and $Y$ represents the reduction in uncertainty of $X$ by knowing $Y$. 

We now return to the original definition given for mutual information. Using the definition of the KL-divergence and conditional probability ($p(x|y) = \frac{p(x,y)}{p(y)}$), we have: 
\begin{align*}
I(X;Y) &= \sum_y p(y) \sum_x p(x|y)\log{\frac{p(x|y)}{p(x)}} \\
       &= \sum_y p(y) D_{KL}(p(x|y)||p(x)) \\
       &= E_Y [D_{KL}(p(x|y)||p(x))]
\end{align*}
Thus we see that the mutual information  is the expectation of the KL-divergence between the probability distribution of $X$ and the probability distribution of $X$ conditioned on $Y$. If these two distributions have a high KL-divergence, then knowing $Y$ provides us a high amount of information regarding the probability distribution of $X$. This is equivalent to saying that the mutual information of $X$ and $Y$ is high.

We will use the chain rule for mutual information: 
\begin{lemma} \label{lem:mutinfochain}
$I(X;Y_1, Y_2, \ldots, Y_k) = I(X;Y_1) + I(X;Y_2|Y_1) + \ldots + I(X;Y_k|Y_{k-1}, \ldots, Y_2, Y_1)$
\end{lemma}

For a proof of the above lemma, see \cite{Cover}. We are now done defining the information theory terms we will need. Lastly, we will need the following lemma which describes the KL-divergence between two Bernoulli random variables with a similar probability of success:
\begin{lemma} \label{lem:berndiv}
$D_{KL}(B_{p \pm \varepsilon}||B_p) = O(\varepsilon^2)$ where $B_p$ is a Bernoulli random variable with probability of success $p$, $\frac{1}{4} \leq p \leq \frac{3}{4}$ and $\varepsilon \leq \frac{1}{8}$.
\end{lemma}

\begin{proof}
Here we prove the $plus$ part of the lemma ($D_{KL}(B_{p + \varepsilon}||B_p) = O(\varepsilon^2)$). The $minus$ part is nearly identical and is thus excluded.
\begin{align*}
D_{KL}(B_{p + \varepsilon}||B_p) &= (p + \varepsilon)\log{\left(\frac{p + \varepsilon}{p}\right)} + (1 - p - \varepsilon)\log{\left(\frac{1 - p - \varepsilon}{1 - p}\right)} \\
&=\log{\left(\frac{1-p-\varepsilon}{1-p}\right)} + p\log{\left(\left(\frac{p+\varepsilon}{p}\right)\left(\frac{1-p}{1-p-\varepsilon}\right)\right)} + \\ &\qquad\qquad\qquad\varepsilon\log{\left(\left(\frac{p+\varepsilon}{p}\right)\left(\frac{1-p}{1-p-\varepsilon}\right)\right)} \\
&=\log{\left(\frac{1-p-\varepsilon}{1-p}\right)} + \log{\left(\left(\frac{p+\varepsilon}{p}\right)^p\left(\frac{1-p-\varepsilon}{1-p}\right)^{-p}\right)} + \\ &\qquad\qquad\qquad\varepsilon\log{\left(\left(\frac{p-p^2-p\varepsilon+\varepsilon}{p-p^2-p\varepsilon}\right)\right)} \\
&=\log{\left(\left(1+\frac{\varepsilon}{p}\right)^p \left(1-\frac{\varepsilon}{1-p}\right)^{1-p}\right)} + \varepsilon\log{\left(1+\frac{\varepsilon}{p(1-p-\varepsilon)}\right)} \\
\intertext{Use the inequalities $1+x \leq e^x$ and $1-x \leq e^{-x}$:}
D_{KL}(B_{p + \varepsilon}||B_p) &\leq \log{\left(e^{\frac{\varepsilon}{p}p} e^{\frac{-\varepsilon}{1-p} (1-p)}\right)} + \varepsilon\log{e^{\frac{\varepsilon}{p(1-p-\varepsilon)}}} \\
&= \log_2{e^0} + \frac{\varepsilon}{\ln{2}} \frac{\varepsilon}{p(1-p-\varepsilon)} \\
\intertext{since $\frac{1}{4} \leq p \leq \frac{3}{4}$ and $\varepsilon \leq \frac{1}{8}$, $p(1-p-\varepsilon) \geq \frac{3}{4}\left(1-\frac{3}{4}-\frac{1}{8}\right) = \frac{3}{32}$:}
D_{KL}(B_{p + \varepsilon}||B_p) &\leq 0 + \frac{32\varepsilon^2}{3\ln{2}} \\
&= O(\varepsilon^2)  \qedhere
\end{align*} 
\end{proof} 

We are now ready to begin our proof of the lower bound. The approach taken is to show that the information gain from each query is small compared with the total information required to find a certain chosen integer. This will allow us to show that a certain minimum number of queries is required to find each of the $k$ integers.

\begin{lemma} 
The lower bound for the number of queries required to find the $k$ integers  between $1$ and $n$ in the set $S$ with probability $>0.99$, when $8 \leq k\le \sqrt{n}$, is $\Omega(k^3\log{n})$
\end{lemma}
\begin{proof}
We choose our input as follows. Split the integers in the range $[1,n]$ into $k$ equally sized clusters. Call these clusters $G_1, G_2, \ldots, G_k$. Let there be one of the $k$ chosen integers in each such cluster. This integer is chosen uniformly at random from the integers in the cluster. Note that the number of integers in each cluster is $\frac{n}{k}$, which, without loss of generality, we will assume is an integer. See Figure \ref{fig:lowerbound} for a visualization of this.

\begin{figure}[ht]
	\centering
		\includegraphics[scale=0.5]{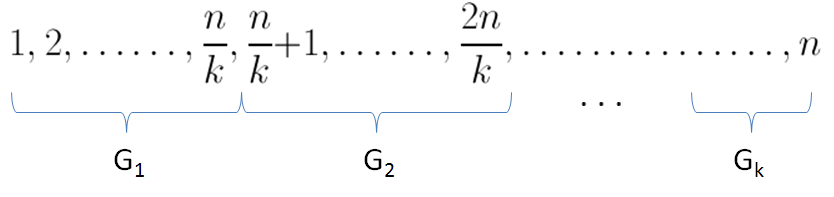}
	\caption{Visualization of our partition of the integers between $1$ and $n$}
	\label{fig:lowerbound}
\end{figure} 

We consider individually a cluster $G_i$ where $\frac{k+4}{4} \leq i \leq \frac{3k}{4}$. Let $L$ be the random variable that represents the chosen integer in $G_i$. Since this number is chosen uniformly at random from $\frac{n}{k}$ elements, the probability of each integer being the chosen integer is $P(x) = \frac{1}{\frac{n}{k}}=\frac{k}{n}$. Therefore, the entropy of $L$ is $H(L) = \sum_{x} P(x)\log{\frac{1}{P(x)}} = \sum_{i = 1}^{\frac{n}{k}} \frac{k}{n} \log{\frac{n}{k}} =  \log{\frac{n}{k}}$. We now define $Q_j$ to be a Bernoulli random variable representing the response to the $j^{th}$ query (i.e. either `$\leq$' or `$>$'). We need to make enough queries so that the information gain relevant to $L$ is close to the entropy of $L$ in order to determine the chosen number in $G_i$ with a high degree of accuracy. This is equivalent to saying that the mutual information between $L$ and the queries made $Q_1, Q_2, \ldots, Q_l$ is at least a constant times the entropy of $L$. Indeed, in the end, we must have determined the point with probability greater than $0.99$. Therefore, conditioned on the queries, most of the mass is concentrated on one point and $H(L|Q_1,\ldots,Q_l) < 0.2 \log{\frac{n}{k}}$. Therefore, $I(L;Q_1,\ldots,Q_l) = H(L) - H(L|Q_1,\ldots,Q_l) = \Om(\log{\frac{n}{k}})$. Thus, we need: 
\begin{equation} \label{eq:mutinfoentropy}
I(L;Q_1, Q_2, \ldots, Q_l) \geq \Om(\log{\frac{n}{k}}),
\end{equation}
where $l$ is the number of queries made. We want to find the minimum $l$ for which this is true. First, we use Lemma \ref{lem:mutinfochain} (chain rule) to write:
\begin{equation} \label{eq:mutinfoLandQ} I(L;Q_1, Q_2, \ldots, Q_l) = I(L;Q_1) + I(L;Q_2|Q_1) + \ldots + I(L;Q_l|Q_{l-1},\ldots,Q_2,Q_1).
\end{equation}

Take one of these terms and recall that we can express mutual information in terms of KL-divergence:
$$I(L;Q_j|Q_{j-1},\ldots,Q_1) = E_Q [D_{KL}(p(Q_j|L,Q_{j-1},\ldots,Q_1)||p(Q_j|Q_{j-1},\ldots,Q_1))]$$ where $1 \leq j \leq l$. Thus, we need to find the KL-divergence of $Q_j|L,Q_{j-1},\ldots,Q_1$ and of $Q_j|Q_{j-1},\ldots,Q_1$. We note that since we chose cluster $G_i$, there are $i-1$ of the $k$ chosen integers that are smaller and $k-i$ of the $k$ numbers that are bigger than any element of $G_i$. Therefore, for both probability distributions, the probability that the response is `$\leq$' is at least $\frac{i-1}{k}$ and the probability that the response is `$>$' is at least $\frac{k-i}{k}$. Therefore, both probability distributions are Bernoulli with probability of success (taking success to be the response `$\leq$') between $\frac{i-1}{k}$ and $1-\frac{k-i}{k} = \frac{i}{k}$. Thus, the difference in probabilities of success of the two distributions is at most $\frac{i}{k} - \frac{i-1}{k} = \frac{1}{k}$. Then if we let $Q_j|L,Q_{j-1},\ldots,Q_1$ be $B_p$ and let $Q_j|Q_{j-1},\ldots,Q_1$ be $B_{p \pm \varepsilon}$, we know $\frac{1}{4} \leq p \leq \frac{3}{4}$ (because $\frac{k+4}{4} \leq i \leq \frac{3k}{4}$) and $0 \leq \varepsilon \leq \frac{1}{k}$ (because this is the maximum difference in probability of success between the two distributions). By lemma \ref{lem:berndiv}, $D_{KL}(p(Q_j|L,Q_{j-1},\ldots,Q_1)||p(Q_j|Q_{j-1},\ldots,Q_1)) = O(\varepsilon^2) = O(\frac{1}{k^2})$. So: $E_Q [D_{KL}(p(Q_j|L,Q_{j-1},\ldots,Q_1)||p(Q_j|Q_{j-1},\ldots,Q_1))] = O(\frac{1}{k^2})$ and we have: $$I(L;Q_j|Q_{j-1},\ldots,Q_1) = O\left(\frac{1}{k^2}\right).$$

Returning to equation \ref{eq:mutinfoLandQ}: 
\begin{align*}
I(L;Q_1, Q_2, \ldots, Q_l) &= \sum_{j=1}^l I(L;Q_j|Q_{j-1},\ldots,Q_1) \\
&= O\left(l \frac{1}{k^2}\right)
\end{align*}

From  \eqref{eq:mutinfoentropy}, we have $O(l \frac{1}{k^2}) \geq \Om(\log{\frac{n}{k}})$ so $$l = \Omega\left(k^2 \log{\frac{n}{k}}\right) = \Omega(k^2 \log{n})$$ since $k \leq \sqrt{n}$. This is the minimum number of queries to find the chosen integer in $G_i$. This holds in total for $\frac{3k}{4} - \frac{k + 4}{4} + 1= \frac{k}{2}$ of the $k$ chosen numbers (this is the number of clusters $G_i$ with $i$ in the range we considered). Note that to find the chosen number in $G_i$, queries made in determining the number within $G_j$ with $j \neq i$ provide no information for determining the number in $G_i$ (as all queries are either bigger or smaller than all the numbers in $G_i$). Then finding $\frac{k}{2}$ of the $k$ chosen numbers requires at least $\Omega\left(\frac{k}{2} k^2 \log{n}\right) = \Omega(k^3 \log{n})$ time. Therefore, finding all $k$ of the chosen numbers requires at least ${\Omega(k^3 \log{n})}$ queries. \qedhere
\end{proof}

\subsection{The lower bound when $k>\sqrt{n}$}

Next we turn our attention to the lower bound in the regime where $k>\sqrt{n}$. We start with the case $\sqrt{n}<k\le n - 2$, as the case $k>n - 2$ is treated very similarly. 
The multi-set $S$ is constructed as follows: we place $k/4$ $1$'s and $k/4$ $n$'s in $S$. Partition the rest of the set $\{1,\ldots,n\}$ into bins $B_1=\{2,3\}$, $B_2=\{4,5\}$, etc. 
For each bin $B_i$ for $i=1,2,\ldots,k/2$, we place exactly one of the elements of $B_i$ in $S$ independently and uniformly at random. We now look at the process of determining which element of $B_i$ has
been selected using the queries. Note that only the query with $Y=2i$ carries any information on which element of $B_i$ has been selected. Thus a set of observations
can be specified by a set of pairs of numbers $\{(l_i,h_i)\}_{i=1}^{k/2}$ where $l_i$ represents the number of times we queried $Y=2i$ and received the `$\le$' answer, 
and $h_i$ represents the number of times we received the `$>$' answer. The probability of each answer is between $1/4$ and $3/4$, and varies by $1/k$ depending on 
whether we selected $2i$ or $2i+1$ in $B_i$. 

When we output the set $S$, we need to make $k/2$ decisions of whether to output $2i$ or $2i+1$ for each $B_i$. Each of these decisions should depend only on the 
values of $(l_i,h_i)$, and should maximize the probability that the output is correct. This can only be done by outputting the maximum likelihood value for each $B_i$. 
More precisely, we should output $2i$ if $\frac{l_i}{l_i+h_i} > \frac{k/4+i-1/2}{k}$, and $2i+1$ otherwise. We are not particularly concerned with these details, but 
only with the probability that our output is wrong. Denote by $\ve_i>0$ the probability that the maximum-likelihood output given $(l_i,h_i)$ is incorrect. 
We first claim that to have a probability of $>0.9$ to be correct in outputting $S$, we must have a bound on the sum of the $\ve_i$'s. 

\begin{claim} \label{cl:1}
If given the values  $\{(l_i,h_i)\}_{i=1}^{k/2}$ the output $S$ is correct with probability $>0.5$, then 
$\sum_{i=1}^{k/2} \ve_i< 1$.
\end{claim}

\begin{proof}
Since the events of being correct on each $B_i$ are independent, the probability of being correct on all $B_i$'s is given by
$$
0.5<\prod_{i=1}^{k/2} (1-\ve_i) < e^{-\sum_{i=1}^{k/2} \ve_i}, 
$$
which implies the statement of the claim. 
\end{proof}

Next, let us denote by $\mu_i$ the {\em a-priori} expected number of `$\le$' responses on $l_i+h_i$ queries, and let $d_i:=|l_i-\mu_i|$ be
the observed deviation from this expected value. Intuitively, the greater this deviation, the greater is our confidence in the answer. 
In fact, it is not hard to formalize this intuition:

\begin{claim} \label{cl:2}
For each $i$, and $k>25$, $\ve_i> e^{-10 d_i/k}/3$. 
\end{claim}

\begin{proof}
Suppose wlog that $l_i>\mu_i$, and thus we are outputting $2i$. Denote $p=\frac{k/4+i-1}{k}$ and $q=\frac{3k/4-i}{k}$. We have by Bayes' rule
\begin{multline*}
\ve_i = Pr[2i+1|(l_i,h_i)] = \frac{Pr[(l_i,h_i)|2i+1]}{2Pr[(l_i,h_i)]} \ge \frac{  Pr[(l_i,h_i)|2i+1]}{2 Pr[(l_i,h_i)|2i]}= \\
\frac{p^{l_i} (q+1/k)^{h_i}}{2(p+1/k)^{l_i} q^{h_i}} = \frac{p^{\mu_i-1} (q+1/k)^{l_i+h_i-\mu_i+1}}{2(p+1/k)^{\mu_i-1} q^{l_i+h_i-\mu_i+1}}
\cdot \frac{p^{l_i-\mu_i+1} (q+1/k)^{\mu_i-l_i-1}}{(p+1/k)^{l_i-\mu_i+1} q^{\mu_i-l_i-1}} = \\
\frac{  Pr[(\mu_i-1,l_i+h_i-\mu_i+1)|2i+1]}{2 Pr[(\mu_i-1,l_i+h_i-\mu_i+1)|2i]} \cdot \left(1-\frac{1/k}{p+1/k}\right)^{d_i+1}\cdot \left(1+\frac{1/k}{q}\right)^{-d_i-1} \ge \\
(1/2)\cdot \left(1-5/k\right)^{2d_i+2} \ge e^{-(5/k)(2d_i+2)}/2 >  e^{-10 d_i/k}/3.
\end{multline*}
The second-to last inequality follows from the fact that the breakdown $(\mu_i-1,l_i+h_i-\mu_i+1)$ is more likely under the selection of $2i+1$ than under the selection of $2i$. 
\end{proof}

Putting Claims~\ref{cl:1} and \ref{cl:2} together we see that assuming the probability that the output $S$ is correct is $>0.5$, we must have
\begin{equation}\label{eq:1}
\sum_{i=1}^{k/2} e^{-10 d_i/k} < 3. 
\end{equation}

\begin{claim}\label{cl:3}
Equation \eqref{eq:1} implies $\sum_{i=1}^{k/2} d_i> \frac{k^2}{40} \ln k$, for $k>40$.  
\end{claim}

\begin{proof}
Denote $\tau_i:=e^{-10 d_i/k}$, and let $f(x):=-\ln x$. The function $f(x)$ is convex, and thus we have
$$
\sum_{i=1}^{k/2} \frac{10 d_i}{k} = \sum_{i=1}^{k/2} f(\tau_i) \ge \frac{k}{2} \cdot f\left(\frac{2}{k}\sum_{i=1}^{k/2} \tau_i\right) > \frac{k}{2} \ln \frac{k}{6}>\frac{k}{4}\ln k,
$$
since $k>40$. This implies the claim. 
\end{proof}

To finish the proof let $D_t$ denote the random variable representing the value of $\sum_{i=1}^{k/2} d_i$ after $t$ queries. Let $Z_t=D_t-\frac{t}{k}$. 
At each time step, a query to $Y=2i$ will on average not change $d_i$ if the element from $B_i$ is not selected for comparison with $Y$. If it is selected, 
it will change $d_i$ by at most $1$. Thus, on average, $D_t$ only grows by at most $\frac{1}{k}$ after each time step. Thus $Z_t$ is a supermartingale. 
Let $T$ be the random variable representing the time at which we stop and output $S$.  By the optional stopping time theorem, we have $E[Z_T]\le 0$, which implies
$E[T]\ge k\cdot E[D_T]$.

If our overall success probability is $>0.75$, it must be the case that with probability $>1/2$ the probability of the output $S$ being correct conditioned 
on the observed $\{(l_i,h_i)\}_{i=1}^{k/2}$ is $>1/2$. Thus by Claims~\ref{cl:1}, \ref{cl:2} and \ref{cl:3}, we have $D_T>\frac{k^2}{40} \ln k$ with probability 
$>1/2$. Thus, 
$$
E[T] \ge k\cdot E[D_T] > k \cdot \frac{1}{2} \cdot \frac{k^2}{40} \ln k = \Om(k^3 \log k), 
$$
completing the proof of the lower bound. 

\begin{remark}
The proof in the regime $k>n - 2$ is very similar. The only difference is that there are $n/2$ bins now, and we'd get $E[D_T]=\Om (k n\log n)$ instead 
of $\Om(k^2 \log n)$, and thus $E[T]=\Om(k^2 n \log n)$. 
\end{remark}

We will now study the case where $k \leq n$.

\section{Optimal upper bounds} \label{sec:algorithm}

As discussed in the previous sections, it is not immediately clear how to make use of the information gained from queries because we do not know which of the $k$ integers the information corresponds to. In this section, we present an algorithm for solving this problem. The algorithm is optimal when the probability of error required is constant (which means its worst case running time matches the lower bound). Our algorithm finds each of the $k$ numbers individually, without attempting to use information gained when finding one integer to find another integer. We first introduce a concept we will use in all our algorithms: 
\begin{mydef}
The $k$-position of an integer $y$ is the number of integers in $S$ that have a value less than or equal to $y$
\end{mydef}
The general technique of the algorithms is to do a binary search for a chosen integer, but repeat each query of the binary search enough times to know the $k$-position of the queried integer. 
A straightforward application of binary search with repeated queries would take $\Om(k^2 \log^2 n)$ queries to find the $k$-position of a number, even with a constant error probability. 
We essentially use the noisy binary search technique of Feige et. al. \cite{Feige} to attain the optimal query complexity.  We start with the following simple lemma:

\begin{lemma} \label{lem:cointoss}
We can find the $k$-position of integer $y$ by making $2k^2\log{\frac{2}{\delta}}$ queries with the probability of being correct being at least $1-\delta$.
\end{lemma}

\begin{proof}
Let $K_y$ be the $k$-position of $y$. We do $m$ queries of $y$ to find $K_y$. For each query $Q_i$, the probability of a response being `$\leq$' or `$>$' is given simply in terms of $K_y$: $$Pr[Q_i = '\leq'] = \frac{K_y}{k}$$ $$Pr[Q_i = '>'] = \frac{1-K_y}{k}$$ because $K_y$ is the number of integers in $S$ less than or equal to $y$ and each such integer is chosen as the $X_i$ for a query with equal probability. We use the analogy that the random variable $Q_i$ is a coin with probability of heads (which represents `$\leq$') being $p = \frac{K_y}{k}$. Given $m$ tosses of the coin, of which $x$ are heads, we can approximate $p$ as: $\hat{p} = \frac{x}{m}$. We need to find the relation between the number of tosses $m$ and the probability of error in this approximation. Using standard concentration bounds \cite{Raginsky}, we see that $m \geq \frac{1}{2\varepsilon^2}\log{\frac{2}{\delta}}$ coin tosses are needed to guarantee that $|\hat{p}-p| \leq \varepsilon$ with error at most $\delta$ (where $\varepsilon > 0$). 

We need to decide on a value for $\varepsilon$. Note that $K_y$ is an integer in the range $[1, k]$ and therefore, $p$ can only take on the values $0, \frac{1}{k}, \frac{2}{k}, \ldots, \frac{k}{k}$. Thus, we need $\varepsilon \leq \frac{1}{2k}$ so then we can always round $\hat{p}$ to the closest $\frac{i}{k}$, where $i \in \mathbb{Z}$ and $0 \leq i \leq k$. Using this in the results from \cite{Raginsky}, we see that $m = 2k^2\log{\frac{2}{\delta}}$ coin tosses are enough to guarantee that we know the correct value of $p$ with probability of error being at most $\delta$. Given $p$, we have $K_y = kp$ so we have the $k$-position of $y$. 
\end{proof}

We note that this immediately lets us solve the problem for $k \geq n$:

\begin{corollary} \label{cor:bigksolution}

When $k \geq n$, there is an $O(k^2 n \log{n})$ algorithm to find all $k$ integers in $S$ with probability $1-n^{-c}$ for all constant $c>0$.

\end{corollary}

\begin{proof}

We find the $k$-position of all $n$ integers in the range $[1, n]$. Given the $k$-position of all $n$ integers, we know how many of the $k$ numbers have each integer value. If the $k$-position of $Y-1$ is $i$ and the $k$-position of $Y$ is $i+j$, we know there are $j$ of the chosen numbers with the value $Y$ (for $1 < Y \leq n$. For $Y = 1$, we know the number of the chosen integers with this value is equal to the $k$-position of $Y$).

To find the $k$-position of an integer with probability of error at most $\delta$, we need to perform $O(k^2 \log{\frac{2}{\delta}})$ queries. If we want the probability of error of the algorithm to be a constant, we need the probability of error of finding the $k$-position of each integer to be at most $\de=n^{-(c+1)}$ so that applying a union bound gives a total probability of error $<n^{-c}$ (since we find the $k$-position of $n$ integers). Thus, to find the $k$-position of each integer we need to perform $O\left(k^2 \log{\frac{2}{\frac{1}{n^{c+1}}}}\right) = O\left(k^2 \log{n}\right)$ queries. Since we do this for $n$ integers, the total number of queries we make is: ${O(k^2 n \log{n})}$.
\end{proof}

We provide an example to illustrate an approach using what we have so far. Suppose $n = 16$, $k = 2$, $S = \{3, 10\}$ and let $m = 2k^2\log{\frac{2}{\delta}}= 8\log{\frac{2}{\delta}}$. We want to find the lowest of the $k$ numbers first. We do a binary search where we repeat each query $m$ times. Our decision at each stage of the binary search is determined by the $k$-position found for the number of that stage. Therefore, we first do $m$ queries of $y=8$. From this, we calculate $K_y = 1$ (note the probability of error for this statement is $\delta$). So there is one of the $k$ numbers below or equal to 8. Next, we do $m$ queries of $y=4$ and find again that $K_y = 1$. When we do $m$ queries of $y=2$, we find that $K_y = 0$. This tells us that none of the $k$ numbers are below or equal to 2. Therefore, we do $m$ queries of $y=3$ and find that $K_y = 1$. If one of the $k$ numbers is less than or equal to 3, but none of them are less than or equal to 2, we conclude that one of the $k$ numbers is 3. We then repeat the same process to find the second of the $k$ numbers. 

However, this approach is problematic because of the constant error each time we find the $k$-position of a number. This flaw is mentioned for a similar algorithm in \cite{Karp}. The number of queries we make is $O(mk\log{n}) = O\left(k^3\log{n}\log{\frac{2}{\delta}}\right)$. Each group of queries of the same $y$ ($m$ of them) give the wrong result with probability $\delta$. Applying a union bound, our overall probability of error ($\Delta$) is $\Delta = k\log{(n)}\delta$. If we want $\Delta$ to be a constant, we need $\delta = \frac{1}{k\log{n}}$ and thus, the number of queries we make is actually $O\left(k^3\log{(n)}\log{(2k\log{n})}\right)$ 

To alleviate this problem, we model our algorithm as a random walk on a tree. In using this technique, we follow \cite{Feige}. In \cite{Feige}, the random walk approach is taken to do a noisy binary search. We use this technique to find each of the chosen $k$ integers, although each step of the random walk is modified to accommodate our lack of information about which of the $k$ integers was chosen in a particular query. We use a binary tree where the leaves are (in order) the integers $1, 2, \ldots, n$. The internal nodes represent intervals that are the union of the leaves in their subtrees. For example, the root node has the interval $[1, n]$ and the left child of the root has the interval $[1, \lfloor\frac{n}{2}\rfloor]$. The tree height is $\log{n}$. Finally, we extend this tree by adding chains of length $m' = O(\log{n})$ to each of the leaf nodes, where the nodes in these chains have the same value as the leaf they are attached to. An example tree with $n = 4$ is shown in Figure \ref{fig:treefig} below.

\begin{figure}[ht]
	\centering
		\includegraphics[scale=0.5]{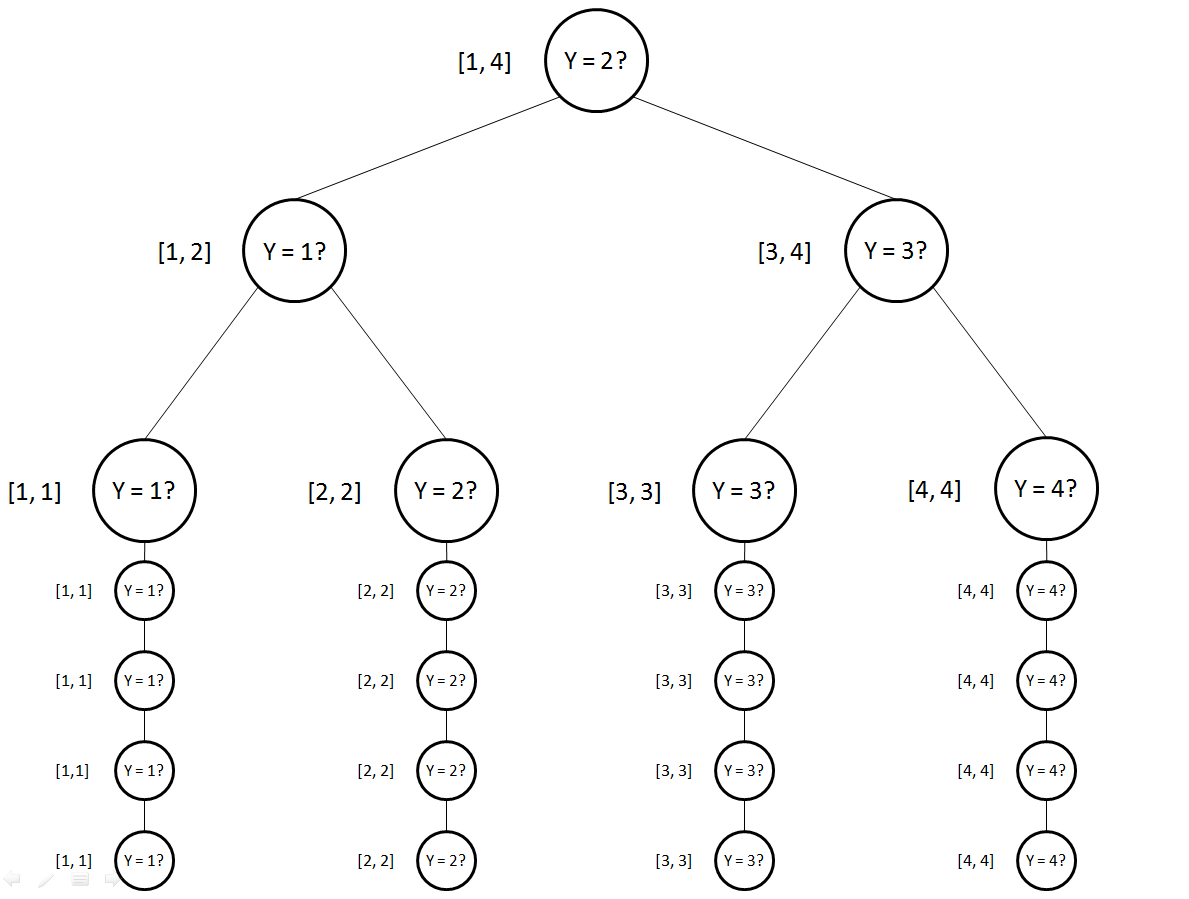}
	\caption{Tree for the random walk with $n = 4$}
	\label{fig:treefig}
\end{figure} 

\subsection{Algorithm} \label{subsec:algdetails}

We discuss an algorithm for finding the $t^{th}$ of the $k$ chosen integers. This algorithm is repeated $k$ times (once for each of the $k$ numbers). Starting at the root, for each node $v$ we take the following two steps:
\begin{enumerate}
	\item We first check whether the $t^{th}$ chosen integer is in the range of the node (call it $[a, b]$). To do this, we find the $k$-position of $a-1$ and $b$ by doing $8k^2$ queries of each of them. If we find that the $k$-position of $a-1$ is at most $t-1$ and the $k$-position of $b$ is at least $t$, then the $t^{th}$ number lies in the range $[a, b]$. Otherwise, we backtrack up the tree to the parent node of $v$ . 
	\item If, according to the first step, the $t^{th}$ number lies in the range $[a, b]$, we do $10k^2$ queries of the middle value of the range of the node (call this $u$ where $u = \lfloor\frac{a+b}{2}\rfloor$). If $v$ is not a leaf (or on a leaf chain) and the $k$-position of $u$ is at most $t-1$, we choose the right child of $v$. If the $k$-position of $u$ is at least $t$, we choose the left child of $v$. If $v$ is a leaf (or on a leaf chain), we go down the chain further regardless of the result of the queries. 
\end{enumerate}

Note that there is a constant probability of error each time we determine the $k$-position of an integer. This leads to a constant probability of choosing the wrong node to go to next. We will analyze this probability shortly. 

The algorithm walks for $m=O(\log{n})$ steps and then stops, where $m < m'$. If it stops on an internal node, the algorithm failed. If it stops on one of the leaf chains (or a leaf node), it outputs the value of the leaf (i.e. declares this value to be the value of the $t^{th}$ of the $k$ numbers). 

The following theorem summarizes our results:

\begin{theorem} \label{thm:algrunningtime}
Our algorithm finds all $k$ integers in $S$ in $O\left(k^3\log{\left(\frac{n}{\delta}\right)}\right)$ time with probability of error at most $\delta$ for $k \leq n$
\end{theorem}

To reach this theorem, we use the following lemma:

\begin{lemma} \label{lem:algcorrectness}
The algorithm finds the correct $t^{th}$ integer in $S$ with the probability of error being at most $e^{-\frac{m}{35}}$, where $m$ is the number of steps in the random walk.
\end{lemma}

\begin{proof}
We need to prove that the algorithm's position on the walk after $m$ steps is the correct leaf chain with high probability. Orient all edges of the tree so they are directed towards the correct leaf chain (and within this leaf chain they are directed down). We can do this because the graph is a tree (there is only one path between every two vertices) and there is only one correct leaf. We can now consider the algorithm's position in the tree as a one dimensional random walk. We let the starting point of the walk be $0$ (the root of the tree), the correct leaf be $R$ steps to the right and any of the wrong leaves be $R$ steps to the left. Note that $R = \log{n}$ (height of the tree). 

We need to find the probabilities of moving left and right in the random walk. We will show that the probability of moving in the correct direction (to the right) is at least $0.7$ at every node. Furthermore, note that the decision made at any node is independent of the previous steps in the random walk. Let $q$ be the probability of going left at any move. This is equivalent to the probability of going along the wrong direction of an edge, which is equivalent to making a mistake somewhere in choosing the next vertex. The probability of incorrectly calculating whether the $t^{th}$ number is in the range $[a, b]$ is at most the probability that we incorrectly calculate the $k$-position of either $a-1$ or $b$. Since we do $8k^2$ queries of each, by Lemma \ref{lem:cointoss} we know that the probability of error in calculating the $k$-position of each is $\delta$ where $2\log{\frac{2}{\delta}} = 8 \Rightarrow \delta = \frac{1}{8}$. So the probability of incorrectly calculating the $k$-position of either $a-1$ or $b$ is at most $1 - \left(\frac{7}{8}\right)^2 = \frac{15}{64}$. Similarly, we do $10k^2$ queries of $u$, so the probability of error is $\delta$ where $2\log{\frac{2}{\delta}} = 10 \Rightarrow \delta = \frac{1}{16}$. Thus, the total probability of error at each node is $\frac{15}{64} + \frac{1}{16} < 0.3$. Therefore, $q < 0.3$ and $p \geq 0.7$, where $p$ is the probability of going to the right (i.e. the correct direction). Figure \ref{fig:randomwalk} illustrates the random walk space. 

\begin{figure}[ht]
	\centering
		\includegraphics[scale=0.4]{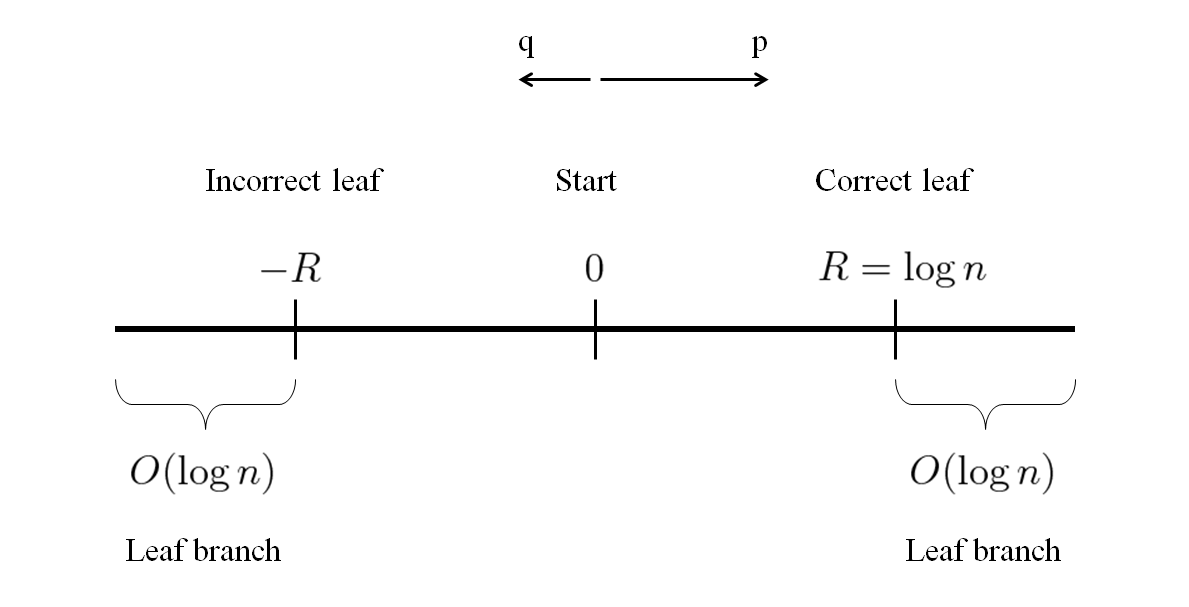}
	\caption{The random walk space}
	\label{fig:randomwalk}
\end{figure} 

For the algorithm to be correct, it must be on or to the right of $R$ after $m$ steps (so it returns the correct integer), otherwise it is wrong. Let $X$ be the random variable denoting the number of moves to the right made after $m$ moves. Then $m-X$ is the number of moves to the left. Therefore, the algorithm is correct if $X - (m - X) = 2X - m \geq R$. This is equivalent to the condition that $X \geq \frac{R + m}{2}$. Then the probability that the algorithm is correct is $Pr[X \geq \frac{R + m}{2}] = 1 - Pr[X < \frac{R + m}{2}]$ and $Pr[X < \frac{R + m}{2}]$ is the probability of error we want to bound. To find $E[X]$, let $X_i$ be an indicator random variable that is 1 if the algorithm moves to the right on the $i^{th}$ move and 0 otherwise. Note that $Pr[X_i = 1] = p \Rightarrow E[X_i] = p$. Therefore, $E[X] = E[X_1] + E[X_2] + \ldots + E[X_m] = pm$ by linearity of expectation. We want to use a Chernoff bound to bound the probability of error, so we need to find a $\delta$ such that:
\begin{align*}
&\frac{m+R}{2} = (1-\delta)pm \\
\Rightarrow \, &1-\delta = \frac{m + R}{2pm} \\
\Rightarrow \,&\delta = \frac{2pm - m - R}{2pm} 
\end{align*}
Note that $0 \leq \delta \leq 1$ because $0 \leq 2pm - m - R \leq 2pm$. Since each step of the random walk is independent of the other steps (i.e. $X_i$ is independent of $X_j$ for $i \neq j$), we can use the Chernoff bound (\cite{Lehman}):
\begin{align*}
Pr[X < \frac{m + R}{2}] &\leq e^{-\frac{\delta E[X]}{2}} \\
&= e^{-\left(\frac{2pm - m - R}{2pm}\right)^2 \frac{pm}{2}} \\
&= e^{-\frac{\left(2pm - m - R\right)^2}{8pm}} 
\end{align*}
Recall that $p \geq 0.7$ and set $m = x\log{n}$, where $x$ is a constant. Then $\frac{(2pm - m - R)^2}{8pm} \geq \frac{\left(1.4x - x - 1\right)^2}{5.6x}\log{n}$. We want to write this as $\frac{m}{d}$ where $d$ is a constant. Then $d = \frac{x}{\frac{\left(0.4x - 1\right)^2}{5.6x}}$. Note that as $x$ increases, $d$ decreases to some asymptotic value: 
\begin{align*}
\lim_{x\to\infty} \frac{x}{\frac{\left(0.4x - 1\right)^2}{5.6x}} &= \lim_{x\to\infty} \frac{5.6 x^2}{\left(0.4x - 1\right)^2} \\
&= \lim_{x\to\infty}\frac{5.6}{\left(0.4 - \frac{1}{x}\right)^2} \\
&= \frac{5.6}{0.4^2} = 35
\end{align*}

Then we have that $\frac{\left(2pm - m - R\right)^2}{8pm} \geq \frac{m}{35}$. Therefore, $$Pr[X < \frac{m + R}{2}] \leq e^{-\frac{m}{35}}.$$ Thus, we have bounded the probability of error as required.
\end{proof}

We apply Lemma \ref{lem:algcorrectness} to prove the bound on the full algorithm.  
Even though our lower bounds works when the error probability is constant, the algorithm applies even when the error is very small ($n^{-O(1)}$).  We are now ready to present the proof for Theorem \ref{thm:algrunningtime}.

\begin{proof}
We prove separately the cases when $\delta \geq \frac{1}{n}$ and when $\delta < \frac{1}{n}$. In the first case, we set $m = 70\log{n}$. By Lemma \ref{lem:algcorrectness}, the probability of not finding the correct $t^{th}$ number is at most $e^{-\frac{70\log{n}}{35}} = e^{\ln{n^{-2 / \ln{2}}}} < \frac{1}{n^2}$. Applying a union bound of this over the $k$ numbers we need to find, the probability of error is at most $\frac{k}{n^2} \leq \frac{\sqrt{n}}{n^2} = \frac{1}{n^{1.5}}$ because $k \leq \sqrt{n}$. Since $\frac{1}{n^{1.5}} < \frac{1}{n} \leq \delta$, the probability of error is bounded as required. So we need in total $70k\log{n}$ steps of the random walk algorithm. Recall that each such step takes $O(k^2)$ queries. Therefore, in total, we have a running time of $O(70k^3\log{n}) = O(k^3\log{n}) = O(k^3\log{\frac{n}{\delta}})$ since $\delta < 1$. 

We now consider the case when $\delta < \frac{1}{n}$. Set $m = 70\log{\frac{1}{\delta}}$. The probability of not finding the correct $t^{th}$ number is at most $e^{-\frac{70\log{\frac{1}{\delta}}}{35}} = e^{-\frac{2}{\ln{2}} \ln{\frac{1}{\delta}}} < \delta^2$ by Lemma \ref{lem:algcorrectness} (and that $\delta < 1$). Applying a union bound over the $k$ numbers we need to find, the overall probability of error is $k\delta^2 < n\delta^2 < \delta$ as required. Thus, we need $O\left(k\log{\frac{1}{\delta}}\right)$ steps in the random walk, where each consists of $O(k^2)$ queries. Therefore, the total running time is $O\left(k^3\log{\frac{1}{\delta}}\right) = O\left(k^3\log{\frac{n}{\delta}}\right)$.
\end{proof}

%
%
%

\newpage
\appendix

\end{document}